\def\year{2019}\relax
\documentclass[letterpaper]{article} 
\usepackage{aaai19}  
\usepackage{times}  
\usepackage{helvet}  
\usepackage{courier}  
\usepackage{url}  
\usepackage{graphicx}  
\usepackage{amsmath}
\usepackage{amssymb}
\usepackage{amsthm,color}
\usepackage{comment}

\usepackage{thmtools}
\usepackage{thm-restate}

\newtheorem{theorem}{Theorem}
\newtheorem{definition}{Definition}
\newtheorem{lemma}{Lemma}

\begin{document}
%
\title{High Dimensional Clustering with $r$-nets}

\author{Georgia Avarikioti, Alain Ryser, Yuyi Wang, Roger Wattenhofer\\ETH Zurich, Switzerland\\\{zetavar,aryser,yuwang,wattenhofer\}@ethz.ch}
\maketitle
\begin{abstract}

Clustering, a fundamental task in data science and machine learning, groups a set of objects in such a way that objects in the same cluster are closer to each other than to those in other clusters. 
In this paper, we consider a well-known structure, so-called $r$-nets, which rigorously captures the properties of clustering. 
We devise algorithms that improve the run-time of approximating $r$-nets in high-dimensional spaces with $\ell_1$ and $\ell_2$ metrics from 
$\tilde{O}(dn^{2-\Theta(\sqrt{\epsilon})})$ to $\tilde{O}(dn + n^{2-\alpha})$, where $\alpha = \Omega({\epsilon^{1/3}}/{\log(1/\epsilon)})$. 
These algorithms are also used to improve a framework that provides approximate solutions to other high dimensional distance problems. Using this framework, several important related problems can also be solved efficiently, e.g., $(1+\epsilon)$-approximate $k$th-nearest neighbor distance, $(4+\epsilon)$-approximate Min-Max clustering, $(4+\epsilon)$-approximate $k$-center clustering. In addition, we build an algorithm that $(1+\epsilon)$-approximates greedy permutations in time $\tilde{O}((dn + n^{2-\alpha}) \cdot \log{\Phi})$ where $\Phi$ is the spread of the input. This algorithm is used to $(2+\epsilon)$-approximate $k$-center with the same time complexity.
\end{abstract}
%
%

\section{Introduction}
Clustering aims at grouping together similar objects, where each object is often represented as a point in a high dimensional space. Clustering is considered to be a cornerstone problem in data science and machine learning, and as a consequence there exist multiple clustering variants. 
For instance, while each cluster may just be represented as a set of points, it is often advantageous to select one point of the data set as a \emph{prototype} for each cluster.


A significant formal representation of such a prototype clustering is known as \emph{$r$-nets}.
Given a large set of $n$ data points in $d$-dimensional space, 
an $r$-net is a subset (the prototypes) of these data points. 
This subset needs to fulfill two properties: First, balls of radius $r$ around each of the prototypes need to contain every point of the whole data set (covering). Second, we must ensure that the prototypes are well separated, i.e., no ball contains the center of any other ball (packing). \emph{Approximate $r$-nets} lift the covering constraint a tiny bit by allowing  balls to have a slightly larger radius than $r$, while preserving the packing property, i.e., any two prototypes still need to have at least distance $r$. 

Throughout this paper, we assume data sets to be large and high dimensional. We therefore assume the number of features $d$ of each object to be non-constant. This leads to interesting and important problems, as this assumption forces us to think about algorithms whose runtime is sub-exponential (preferably linear) in the number of features $d$. In addition, we want our runtime to be sub-quadratic in the size $n$ of our data. In this paper we lay theoretical groundwork, by showing improved algorithms on the approximate $r$-net problem and applications thereof. 

\subsection{Related Work}
There is not a unique best clustering criterion,
hence many methods \cite{estivill2002so} are proposed to solve the clustering problem for different applications (e.g., \cite{sibson1973slink,defays1977efficient,lloyd1982least,kriegel2011density}), which makes it difficult to systematically analyze clustering algorithms. 

In our paper we will make use of so-called polynomial threshold functions (PTF), 
a powerful tool developed by \cite{alman2016polynomial}. PTFs are distributions of polynomials that can efficiently evaluate certain types of Boolean formulas with some probability. They are mainly used to solve problems in circuit theory, but were also used to develop new algorithms for other problems such as approximate all nearest neighbors or approximate minimum spanning tree in Hamming, $\ell_1$ and $\ell_2$ spaces. In the following, we employ this method to develop an algorithm that computes approximate $r$-nets. 

The algorithmic framework Net \& Prune, was developed by \cite{PR14}. It is able to solve so called nice distance problems, when provided with a suitable data structure for the problem. These data structures are often constructed by exploiting $r$-nets. A major drawback of the framework is its restriction to a constant number of features. Consequentially, this framework was later extended by \cite{AEKP} to also solve higher dimensional cases. The algorithm, constructed in this paper, yields an immediate improvement on this framework, as the construction of the framework is based around approximate $r$-nets. We also present various of the previously mentioned data structures that we plug into the framework to solve high dimensional distance optimization problems.

Recent work by \cite{EHS15} suggests a way of constructing approximate greedy permutations with approximate $r$-nets. Greedy permutations imply an ordering of the data, which provide a solution to $2$-approximate $k$-center clustering as shown by \cite{Gon85}. We present a similar construction, by applying approximate greedy permutations.

An approach on hierarchical clustering was presented by \cite{Dasgupta_2002_Hierarchical_clustering}. They construct an algorithm based on furthest first traversal, which is essentially building a greedy permutation and then traversing the permutation in order.

In \cite{fern_2003_random_projections} they present how random projections, in practice, can be applied to reduce the dimension of given data. We later employ a similar approach, namely random projections to lines, to reduce the approximate $r$-net problem with $\ell_1$ metrics to a low-dimensional subspace.

\subsection{Our Contribution}
This paper presents new theoretical results on the construction of $(1+\epsilon)$-approximate $r$-nets, improving the previous upper bound of $\tilde{O}(dn^{2-\Theta(\sqrt{\epsilon})})$ by \cite{AEKP}. 
We denote $n$ as the number of data points, $d$ the dimension of the data and $\alpha = \Omega(\epsilon^{\frac{1}{3}}/\log(\frac{1}{\epsilon}))$ for an arbitrary error parameter $\epsilon$.
The algorithm builds approximate $r$-nets in Hamming, $\ell_1$ and $\ell_2$ spaces, running in $\tilde{O}(n^{2-\alpha}+n^{1.7+\alpha}d)$\footnote{The $\tilde{O}$ notation throughout the paper hides logarithmic factors in $n$ and polynomial terms in $\frac{1}{\epsilon}$} time in both Hamming and Euclidean space.  

We also modify our algorithm to yield an improvement on the Net \& Prune framework of \cite{AEKP}. 
Supplying the framework with certain data structures, which are created using approximate $r$-nets, we derive new algorithms with improved runtime on $(1+\epsilon)$-approximate $k$-smallest nearest neighbor distance, $(4+\epsilon)$-approximate Min-Max Clustering, introduced in \cite{PR14}, and $(4+\epsilon)$-approximate $k$-center clustering. 
These algorithms run in $\tilde{O}(dn+n^{2-\alpha})$ for data sets in $\ell_1$ or $\ell_2$ spaces. 
With the exception of approximate $k$-smallest nearest neighbor, this is, to our knowledge, the first time this framework is used to solve these problems in high dimensional spaces. 
We later also design a new algorithm to $(2+\epsilon)$-approximate $k$-center clustering, by deriving an improved version of the algorithm for $(1+\epsilon)$-approximate greedy permutations in \cite{EHS15}. 
Both of these algorithms have a runtime of $\tilde{O}((dn+ n^{2-\alpha})\log{\Phi})$, where $\Phi$ denotes the spread of the data. We define the spread of a dataset as the fraction of the diameter over the shortest distance of the graph. 

The omitted proofs can be found in the appendix.
 
\section{Approximate $r$-nets}


In this section, we present an algorithm that builds approximate $r$-nets in $\ell_1$ and $\ell_2$ spaces. To that end, we first derive an algorithm, that constructs approximate $r$-nets in Hamming space. We later show how to reduce the problem from $\ell_1$ or $\ell_2$ to Hamming space.

\subsection{Approximate $r$-net in Hamming Space}
Building approximate $r$-nets in Euclidean space is computationally expensive. Therefore, we initially restrict ourselves to datapoints on the vertices of a high dimensional hypercube. The distance between any two datapoints is then measured by the Hamming distance. In the following, we define the notion of approximate $r$-nets in this metric space, where the error is additive instead of multiplicative. 

\begin{definition}\label{def r-net}
Given a point set $X\subset \{0,1\}^d$, a radius
$r\in \mathbb{R}$, an approximation parameter $\epsilon > 0$ and the Hamming distance denoted as $\|\cdot\|_1$, an approximate $r$-net of $X$ with additive error $\epsilon$ is as subset $C\subset X$ such that the following properties hold:
\begin{enumerate}
\item (packing) For every $p,q \in C$, $p \neq q$, it holds that 
\begin{equation*}
\begin{split}
\| p- q\|_1\geq r
\end{split}
\end{equation*}

\item (covering) For every $p \in X$, there exists a $q\in C$, s. t. 
\begin{equation*}
\begin{split}
\|p - q\|_1 &\leq r + \epsilon d \qquad \text{(additive error)}
\end{split}
\end{equation*}
\end{enumerate}
\end{definition}
To construct approximate $r$-net we employ Probabilistic Polynomial Threshold Functions, a tool introduced in \cite{alman2016polynomial}. To effectively apply this technique, we require a sparse dataset, meaning that we assume that most of the points are further than $r$ from each other. To that end, we present a technique that sparsifies the data in advance without losing meaningful data for our problem.

\subsubsection{Sparsification}
To sparsify our data, we apply brute force to build part of the approximate $r$-net. Intuitively, we randomly pick a center point from our dataset and then remove every point that is closer then $ r + \epsilon d $ from the center, by checking every point of the dataset. 
This technique was originally introduced in \cite{AEKP}. The proof of Theorem \ref{brute force} closely follows this work. 

\begin{restatable}[]{theorem}{bruteforce}\label{brute force}
Given $X \subset\{0,1\}^d$, $|X| = n$, $\epsilon > 0$, the Hamming distance which we denote as $\|\cdot\|_1$ and some distance $r\in\mathbb{R}$, we can compute a set $X'\subset X$ with 
\begin{equation*}
\begin{split}
Pr[Y\leq n^{1.7}]\geq 1-n^{-0.2}
\end{split}
\end{equation*}
and a partial r-net $C$ of $X\setminus X'$, where \\
$Y:=|\{\{i,j\}|x_i,x_j\in X',\|x_i-x_j\|_1\leq r+\epsilon d\}|$ the number of points with close distance to each other, in time $O(dn^{1.5})$. 
\end{restatable}

\subsubsection{Distance Matrix}
Next we introduce a tool, called \emph{distance matrix}, to approximate $r$-nets. 
To construct a distance matrix, we partition the dataset into disjoint sets of equal size.
The rows of the matrix correspond to partitions and the columns to points of the dataset. Each entry holds a value which indicates if any of the points in a partition (row) is at most $r+\epsilon d$ close to a data point (column). We use Probabilistic Polynomial Threshold Functions, formally defined below, to construct a matrix with such indicator values. 

\begin{definition}[\cite{alman2016polynomial}]\label{def prob. polynomial}
If $f$ is a Boolean function on $n$ variables, and $R$ is a ring, a \textit{probabilistic polynomial for $f$ with error $\frac{1}{s}$ and degree $d$} is a distribution $\mathcal{D}$ of degree-$d$ polynomials over $R$ such that $\{0,1\}^n, Pr_{p\sim\mathcal{D}}[p(x) = f(x)]\geq 1-\frac{1}{s}$.
\end{definition}

The main building block to construct the distance matrix is Theorem \ref{prob. PTF for OR}, which uses the fact that each entry of the distance matrix can be expressed as a Boolean formula.

\begin{theorem}[\cite{alman2016polynomial}]\label{prob. PTF for OR}
Given $d,s,t,\epsilon$, we can construct a probabilistic polynomial $\tilde{P}:\{0,1\}^{ns}\rightarrow\mathbb{R}$ of degree at most $\Delta := O((\frac{1}{\epsilon})^{\frac{1}{3}}\log(s))$ with at most $s\cdot {n\choose\Delta}$, such that:
\begin{enumerate}
\item If $\bigvee_{i=1}^s[\sum_{j=1}^n x_{ij}\geq t]$ is false, then $|\tilde{P}(x_{11},...,x_{1n},...,x_{s1},...,x_{sn})| \leq s$ with probability at least $\frac{2}{3}$;
\item If $\bigvee_{i=1}^s[\sum_{j=1}^n x_{ij}\geq t + \epsilon n]$ is true, then $\tilde{P}(x_{11},...,x_{1n},...,x_{s1},...,x_{sn}) > 2s$ with probability at least $\frac{2}{3}$.
\end{enumerate}
\end{theorem}

Before we show how to construct the distance matrix for a given dataset, we cite the following Lemma by \cite{Cop}, on rectangular matrix multiplication.

\begin{lemma}[\cite{Cop}]\label{matmult copper}
For all sufficiently large $N$, and $\alpha \leq .172$, multiplication of an $N\times N^\alpha$ matrix with an $N^\alpha\times N$ matrix can be done in $N^2poly(\log{N})$ arithmetic operations, over any field with $O(2^{poly(log{N})})$ elements.\footnote{A proof can be found in the Appendix of \cite{Wil14}}
\end{lemma}

Next, we present how to build the distance matrix, combining fast matrix multiplication and Probabilistic Polynomial Threshold Functions.

\begin{restatable}[]{theorem}{distancematrix}\label{distance matrix}
Let $X$ be a set of $n$ points in $\{0,1\}^d$, a radius $r \in\mathbb{R}$, some $\epsilon \gg \frac{\log^6(d\log{n})}{\log^3{n}}$, $\alpha = \Omega(\frac{\epsilon^{\frac{1}{3}}}{\log(\frac{d}{\epsilon\log{n}})})$ and let \\
$\|\cdot\|_1$ denote the Hamming distance. There exists an algorithm that computes, with high probability, a $n^{1-\alpha}\times n$ matrix $W$ and a partition $S_1, ..., S_{n^{1-\alpha}}$ of X that satisfies the following properties:
\begin{enumerate}
\item For all $i \in [n^{1-\alpha}]$\footnote{By $[k]$ we denote the set $\{1,2,...,k\}$} and $j\in[n]$, if $\min_{p \in S_i}\|x_j- p\|_1\leq r$ then $W_{i,j} > 2|S_i|$.
\item For all $k \in [n^{1-\alpha}]$ and $j\in[n]$, if $\min_{p \in S_i}\|x_j- p\|_1 > r + \epsilon d$, then $|W_{i,j}| \leq |S_i|$
\end{enumerate}
The algorithm runs in \~{O}$(n^{2-\alpha})$.
\end{restatable}

\subsubsection{Building a Net}
Now, we present how we can build an approximate $r$-net for a data set, as in \cite{AEKP}: we first employ the sparsification technique and then build the distance matrix in the sparse dataset where we can search efficiently.
 The running time of building an approximate $r$-net is dominated by the time complexity of the construction of the distance matrix.

\begin{theorem}\label{hamming r-net}
Given $X \subset\{0,1\}^d$ with $|X| = n$, some distance $r\in\mathbb{R}$ and some $\epsilon \gg \frac{\log^6(d\log{n})}{\log^3{n}}$, we can compute a set $C$ that contains the centers of an approximate $r$-net with additive error at most $\epsilon$ with high probability in time $\tilde{O}(n^{2-\alpha}+dn^{1.7+\alpha})$, where $\alpha = \Omega(\frac{\epsilon^{\frac{1}{3}}}{\log(\frac{d}{\epsilon\log{n}})})$.
\end{theorem}

\begin{proof}
We apply Theorem \ref{brute force} to the set $X$ with radius $r$ and error $\epsilon$. This results in the remaining points $X'$, a partial approximate $r$-net $C'$ for $X\setminus X'$ and $Pr[Y \leq n^{1.7}]\geq 1-n^{-0.2}$, where $Y:=|\{\{i,j\}|x_i,x_j\in X',\|x_i-x_j\|_1\leq r+\epsilon d\}|$, in time $O(n^{1.5}d)$.
We then apply Theorem \ref{distance matrix} to assemble the distance matrix $W$ and the partition $S_1,...,S_{n^{1-\alpha}}$ on inputs $X',\epsilon$ and $r$. If we encounter more then $n^{1.7}$ entries $W_{ij}$ where $W_{ij}>2|S_i|$, we restart the algorithm. Since $Pr[Y\leq n^{1.7}]\geq 1-n^{-0.2}$, with high probability we pass this stage in a constant number of runs.

Next, we iterate top down over every column of $W$. For every column $j$, first check if $x_j$ is already deleted. If this is the case we directly skip to the next column. Otherwise set $C = C\cup \{x_j\}$ and delete $x_j$. For every entry in that column such that $W_{i,j}>2|S_i|$, we then delete every $x \in S_i$ where $\|x_j-x\|_1 \leq r+\epsilon d$.

As we iterate over every column of $W$, which correspond to every point of $X'$, the net fulfills covering. Since points that are added as centers were not covered within $r$ of a center by previous iterations, $C$ also fulfills packing. Thus $C$ now contains the net points of an approximate $r$-net of $X'$.

By Theorem \ref{distance matrix}, building the distance matrix takes $\tilde{O}(n^{2-\alpha})$ time and iterating over every entry of $W$ takes $\tilde{O}(n^{2-\alpha})$ time as well. For at most $n^{1.7}$ of these entries, we check the distance between points in a set $S_i$ and the point of the current column which takes another $O(n^{1.7+\alpha}d)$. The runtime is thus as stated.
\end{proof}

\subsection{Approximate r-nets in Euclidean Space}
In this subsection, we reduce the problem of computing approximate $r$-nets from Euclidean to Hamming space. Then, we apply Theorem \ref{hamming r-net} to compute approximate $r$-nets in Euclidean space. 

We distinguish between $\ell_1$ and $\ell_2$ metrics. Specifically, we first show how to map the dataset from $\mathbb{R}^d$ space with $\ell_1$ metric to Hamming space. Then, we present a reduction from $\ell_2$ to $\ell_1$ Euclidean space. Note that the error in the original Euclidean space is multiplicative, while in Hamming space additive. Although the proofs of both Theorem \ref{r-net in l1} and \ref{l2 to l1} use the same mappings as in \cite{alman2016polynomial}, the problem of computing $r$-nets is different and more general than finding the closest pair, thus we cannot directly cite their results.

\subsubsection{$\ell_1$ case}
To reduce the problem of computing approximate $r$-nets from the Euclidean to the Hamming space (with $\ell_1$ metric) we use  Locality Sensitive Hashing (LSH), formally defined below.


\begin{definition}\label{sensitive hashing}
Let $r,c \in\mathbb{R}$ and $p_1,p_2\in[0,1]$ where $p_1>p_2$. A distribution $\mathcal{D}$ of hash functions is called \textit{$(r,cr,p_1,p_2)$-sensitive}, if, for a metric space $V$ under norm $\|\cdot\|$, a hash function $h$ randomly drawn from $\mathcal{D}$ satisfies the following conditions for any points $x,y\in V$:
\begin{enumerate}
\item $\|x-y\| \leq r \Rightarrow Pr[h(x) = h(y)] \geq p_1$
\item $\|x-y\| \geq cr \Rightarrow Pr[h(x) = h(y)] \leq p_2$
\end{enumerate}
We call hashing methods that exhibit these properties \textit{locality sensitive hash functions (LSH)}.
\end{definition}

Now, we show how to compute approximate $r$-nets in Euclidean space under the $\ell_1$ norm, by employing a specific instance of LSH functions. 

\begin{restatable}[]{theorem}{ellone}\label{r-net in l1}
For a set of input points $X\subset\mathbb{R}^d$, some radius $r\in\mathbb{R}$ and some error $\epsilon\gg \frac{\log^6(\log{n})}{\log^3{n}}$, with high probability, we can construct a $(1+\epsilon)$-approximate $r$-net under $\ell_1$ euclidean norm $\|\cdot\|$ in time $\tilde{O}(dn + n^{2-\alpha})$ where $\alpha = \Omega(\epsilon^{\frac{1}{3}}/\log(\frac{1}{\epsilon}))$.
\end{restatable}

\begin{proof}
The following inequalities, are the same as the ones derived in \cite{alman2016polynomial} for their algorithm that finds all nearest/furthest neighbors of a set.
First apply a variant of locality sensitive hashing, to map points from $\ell_1$ to Hamming space. For each point $p\in X$ and $i\in\{1,...,k\}$, $k=O(\epsilon^{-2}\log{n})$, we define hash functions $h_i =\{h_{i1}(p),...,h_{id}(p)\}$, where $h_{ij} = \left \lfloor{\frac{p_{a_{ij}}+b_{ij}}{2r}}\right \rfloor$, $a_{ij}\in\{1,...,d\}$ and $b_{ij}\in[0,2r)$ sampled independently uniformly at random. For each value $h_i(p)$, define $f_i(p)=0$ with probability $\frac{1}{2}$ or $f_i(p)=1$ with probability $\frac{1}{2}$. We then define a new point in Hamming space as $f(p) = (f_1(p),...,f_k(p))$. We have that for any $p,q\in X$, 
$Pr[h_{ij}(p)\neq h_{ij}(q)] = \frac{1}{d}\sum_{a=1}^d\min\{\frac{|p_a-q_a|}{2r},1\}$ and \\
$Pr[f_i(p) \neq f_i(q)] \\
= Pr[f_i(p) \neq f_i(q)|h_i(p)=h_i(q)]Pr[h_i(p)=h_i(q)] \\
+ Pr[f_i(p) \neq f_i(q)|h_i(p)\neq h_i(q)]Pr[h_i(p)\neq h_i(q)]\\
= 0 + \frac{1}{2}Pr[h_i(p)\neq h_i(q)]= \frac{1}{2} Pr[\bigvee_{j=1}^d h_{ij}(p) \neq h_{ij}(q)]\\
= \frac{1}{2} (1-\prod_{j=1}^d Pr[h_{ij}(p) = h_{ij}(q)])\\
= \frac{1}{2} (1-\prod_{j=1}^d (1 - Pr[h_{ij}(p) \neq h_{ij}(q)]))$


Thus, the following properties hold for any $p,q\in X$:
\begin{enumerate}
\item If $\|p-q\|\leq r$ then $Pr[h_{ij}(p) \neq h_{ij}(q)]\leq\frac{1}{2d}$ and thus \\
$Pr[f_i(p) \neq f_i(q)] \leq  \frac{1}{2} (1-(1-\frac{1}{2d})^d) := \alpha_0$
\item If $\|p-q\|\geq (1+\epsilon)r$ then $Pr[h_{ij}(p) \neq h_{ij}(q)]\geq\frac{1+\epsilon}{2d}$ and thus\\
$Pr[f_i(p) \neq f_i(q)] \leq  \frac{1}{2} (1-(1-\frac{1+\epsilon}{2d})^d) := \alpha_1$
\end{enumerate}
Then it follows that $\alpha_1-\alpha_0 = \Omega(\epsilon)$. By applying a Chernoff bound, we derive the following:
\begin{enumerate}
\item If $\|p-q\|\leq r$ then $\mathbb{E}[\|f(p) - f(q)\|] = \sum_{i=1}^k Pr[\|f_i(p)-f_i(q)\|]\leq k\alpha_0$ and thus $Pr[\|f(p) - f(q)\| \leq \alpha_0 k + O(\sqrt[]{k\log{n}}):= A_0]\geq 1-O(\frac{1}{n})$ 
\item If $\|p-q\|\geq (1+\epsilon)r$ then $\mathbb{E}[\|f(p) - f(q)\|] = \sum_{i=1}^k Pr[\|f_i(p)-f_i(q)\|] \geq k\alpha_1$ and thus $Pr[\|f(p) - f(q)\| \geq \alpha_1 k - O(\sqrt[]{k\log{n}}):= A_1]  \geq 1-O(\frac{1}{n})$
\end{enumerate}
As we know that $\alpha_1-\alpha_0 = \Omega(\epsilon)$ it is easy to see that $A_1-A_0 = k(\alpha_1-\alpha_0)-O(\sqrt[]{k\log{n}}) = \Omega(k\epsilon)$. For the new set of points $X':=f(X)$, we construct an approximate $r$-net with additive error $\Omega(\epsilon)$, which yields the center points of an approximate $r$-net of the original points with multiplicative error $(1+\epsilon)$. \\
We hence apply Theorem \ref{hamming r-net} on inputs $X = X'$, $d=k$, $\epsilon = \Omega(\epsilon)$ and $r=A_0$. This gives us the centers $C$ of an approximate $r$-net for $X'$ in time $\tilde{O}(n^{2-\alpha}+n^{1.7+\alpha})$ where $\alpha = \Omega(\epsilon^{\frac{1}{3}}/\log(\frac{1}{\epsilon}))$. The points that get mapped to the net points in $C$ are then the centers of a $(1+\epsilon)$-approximate $r$-net of the points in $X$ under $\ell_1$ metrics with high probability. Applying this version of locality sensitive hashing to $X'$ takes $O(dkn)=\tilde{O}(dn)$ time, which leads to the runtime as stated.
\end{proof}

\subsubsection{$\ell_2$ case}
Firstly, we reduce the dimension of the dataset using the Fast Johnson Lindenstrauss Transform \cite{johnson1982embeddingl,johnson1984extensions}. Specifically, we use a variant that allows us to map $\ell_2$ points to $\ell_1$ points, while preserving a slightly perturbed all pair distance under the respective norm, as  for example seen in \cite{MAT}. Thus, we can construct approximate r-nets in the general Euclidean space, as formally stated below.

\begin{restatable}[]{theorem}{elltwo}\label{l2 to l1}
For set of input points $X\subset\mathbb{R}^d$, some radius $r\in\mathbb{R}$, some error $\epsilon\gg \frac{\log^6(\log{n})}{\log^3{n}}$,  with high probability, we can construct a $(1+\epsilon)$-approximate $r$-net under $\ell_2$ euclidean norm $\|\cdot\|$ in time $\tilde{O}(dn^{1.7+\alpha} + n^{2-\alpha})$ where $\alpha = \Omega(\epsilon^{\frac{1}{3}}/\log(\frac{1}{\epsilon}))$.
\end{restatable}

\section{Applications}

In the following section, we present applications for the algorithms we presented in the previous section. To that end, we exhibit an improvement on a framework called Net \& Prune. Net \& Prune was invented by \cite{PR14} for low dimensional applications. An extended version of the framework, that is efficient  in higher dimensional datasets, was later presented by \cite{AEKP}. In what follows, we apply the approximate $r$-net algorithm to immediately improve the high dimensional framework. We then present various applications, that depend on approximate $r$-nets and the framework.

\subsection{Net \& Prune Framework}
Net \& Prune mainly consists of two algorithms, {\fontfamily{cmtt}\selectfont ApprxNet} and {\fontfamily{cmtt}\selectfont DelFar}, which are alternatively called by the framework, and a data structure that is specific to the problem we want to solve. When supplied with these, the framework returns an interval with constant spread, which is guaranteed to contain the optimal solution to the objective of the desired problem. To improve the framework, we first improve these two algorithms. {\fontfamily{cmtt}\selectfont ApprxNet} computes an approximate $r$-net for a given point set and {\fontfamily{cmtt}\selectfont DelFar} deletes points the isolated points, i.e. the points that do not contain any other point in a ball of radius $r$ around them.

As an improvement to {\fontfamily{cmtt}\selectfont ApprxNet}, we refer to Theorem \ref{r-net in l1} and Theorem \ref{l2 to l1}. We now present an algorithm, that yields an improved version of {\fontfamily{cmtt}\selectfont DelFar}:
\begin{theorem}
For a set of points $X$, some error \\$\epsilon\gg \frac{\log^6(\log{n})}{\log^3{n}}$, a radius $r \in\mathbb{R}^d$ and the norm $\|\cdot\|$, that denotes the $\ell_1$ or $\ell_2$ norm, we can construct an algorithm {\fontfamily{cmtt}\selectfont DelFar} that outputs, with high probability, a set $F$, where the following holds:
\begin{enumerate}
\item If for any point $p\in X$ it holds that \\$\forall q \in X, q\neq p , \|p-q\|>(1+\epsilon)r$ then $p\not\in F$
\item If for any point $p\in X$ it holds that \\$\exists q\in X, q\neq p,  \|p-q\|\leq r$ then $p\in F$
\end{enumerate}
We do this in time $\tilde{O}(dn+n^{2-\alpha})$, where $\alpha = \Omega(\frac{\epsilon^{\frac{1}{3}}}{\log(1/\epsilon)})$.
\end{theorem}
\begin{proof}
We prove the Theorem for the $\ell_1$ metric. For an $\ell_2$ instance we can simply apply the mapping of Theorem \ref{l2 to l1} and the proof holds.
Initially, we map the points to Hamming space, applying the techniques described in Theorem \ref{r-net in l1}. During the brute force part of the algorithm we do the following: we delete each point that is covered by a center and then we add both the point and the center to set $F$. We do not add centers to $F$ that do not cover any other points. Later, when traversing the distance matrix, we check each entry that indicates if the partition contains a close point. We calculate the distances between the current point and all points of the partition. We add in set $F$, and then delete, the points that are actually close. We ignore points, where every entry in its column indicate no close points. In the end, we return the set $F$. 
The running time of the algorithm is the same as in Theorem \ref{l2 to l1}, since 
deleting a point after adding it to set $F$ takes constant time.
\end{proof}
The Net \& Prune framework allows us to solve various so called nice distance problems. As presented by \cite{PR14}, the problems solved in the subsections to come, are all proven to be of such kind. One of the properties of such problems is, that there needs to exist a so called $(1+\epsilon)$-decider for that problem. In the following, we denote a formal definition of such deciders.

\begin{definition}[\cite{AEKP}]
Given a function $f: X\rightarrow\mathbb{R}$, we call a decider procedure a $(1+\epsilon)$-decider for $f$, if for any $x\in X$ and $r > 0$, $decider_f(r,X)$ returns one of the following: (i) $f(x) \in [\beta, (1+\epsilon)\beta]$ for some real $\beta$, (ii) $f(x) < r$, or (iii) $f(x) > r$.
\end{definition}

Even though \cite{PR14} presented a decider for each of the problems that follow, the extended framework by \cite{AEKP} requires deciders to be efficient, as otherwise the frameworks runtime does not hold. This is a good opportunity to apply Theorem \ref{r-net in l1} and Theorem \ref{l2 to l1} from the previous section. In the following sections, we employ the theorem below to find constant spread intervals. These contain the solutions to nice distance problems. We then apply deciders to approximate the solution of the problem.

\begin{theorem}[\cite{AEKP}]
For $c\geq 64$, the Net \& Prune algorithm computes in $O(dn^{1.999999})$ time a constant spread interval containing the optimal value $f(X)$, with probability $1-o(1)$.
\end{theorem}

\subsection{$k$th-Smallest Nearest Neighbor Distance}
When having a set of points in high dimensional space, we may be interested in finding the $k$-smallest nearest neighbor distance. This means, when looking at the set of distances to the nearest neighbor of each point, finding the $k$th-smallest of these. Computing this with a naive algorithm takes $O(dn^2)$, which is not suitable in high dimensional space. Alternatively, we are able to build a $(1+\epsilon)$-decider and then apply the Net $\&$ Prune framework to solve the problem. This has previously been done by \cite{AEKP}. Theorem \ref{r-net in l1} and Theorem \ref{l2 to l1} yield immediate improvement on the runtime of the decider, as it is built with help of {\fontfamily{cmtt}\selectfont DelFar}. We thus omit the proof of the following Theorem here. The proof can be found in the supplementary material.

\begin{restatable}[]{theorem}{kthdist}\label{kth nn dist}
For a set of points $X$, $\epsilon\gg \frac{4\log^6(\log{n})}{\log^3{n}}$ and the norm $\|\cdot\|$, that denotes the $\ell_1$ or $\ell_2$ norm, with high probability, we can find the $(1+\epsilon)$-approximate $k$-smallest nearest neighbor distance of $X$ in time $\tilde{O}(dn + n^{2-\alpha})$, where $\alpha = \Omega(\frac{\epsilon^{\frac{1}{3}}}{\log(1/\epsilon)})$.
\end{restatable}

\subsection{Min-Max Clustering}
To understand the following problem, we first define Upward Closed Set Systems and Sketchable Families, as introduced in \cite{PR14}.
\begin{definition}[Upward Closed Set System \& Sketchable Families \cite{PR14}]
Let $P$ be a finite ground set of elements, and let $\mathcal{F}$ be a family of subsets of $P$. Then $(P,\mathcal{F})$ is an \textit{upward closed set system} if for any $X\in \mathcal{F}$ and any $Y\subset P$, such that $X\subset Y$, we have that $Y\in\mathcal{F}$. Such a set system is a \textit{sketchable family}, if for any set $S\subset P$ there exists a constant size \textit{sketch} sk($S$) such that the following hold.
\begin{enumerate}
\item For any $S,T\subset P$ that are disjoint, sk($S\cup T$) can be computed from sk($S$) and sk($T$) in $O(1)$ time. We assume the sketch of a singleton can be computed in $O(1)$ time, and as such the sketch of a set $S\subset P$ can be computed in $O(|S|)$.
\item There is a membership oracle for the set system based on the sketch. That is, there is a procedure orac such that given the sketch of a subset sk($S$), orac returns whether $S\in \mathcal{F}$ or not, in $O(1)$ time.
\end{enumerate}
\end{definition}
Min-Max Clustering is a method of clustering sets of the Upward Closed Set Systems within Sketchable Families under some cost function. The following is a formal definition of Min-Max clustering, as provided by \cite{PR14}.
\begin{definition}[Min-Max Clustering \cite{PR14}]
We are given a sketchable family $(P,\mathcal{F})$, and a cost function $g:2^P\rightarrow\mathbb{R}^+$. We are interested in finding disjoint sets $S_1,...,S_m\in \mathcal{F}$, such that $(i) \bigcup_{i=1}^{m} S_i = P$, and $(ii) \max_i g(S_i)$ is minimized. We will refer to the partition realizing the minimum as the \textit{optimal clustering} of $P$.
\end{definition}
We later resort to the following Lemma when building a $(1+\epsilon)$-decider for a concrete instance of Min-Max Clustering.
\begin{lemma}\label{2r balls}
Given a set of $n$ points $X\subset\mathbb{R}^d$, a radius $r\in\mathbb{R}$, some error parameter $\epsilon\gg \frac{\log^6(\log{n})}{\log^3{n}}$, the norm $\|\cdot\|$, that denotes the $\ell_1$ or $\ell_2$ norm, and a set $C\subset X$ s.t. $\forall x,y\in C,  \|x-y\|\geq 2r(1+\epsilon)$,  with high probability, we can return sets $P_i$, such that $\forall c_i \in C,\forall x \in P_i, \|c_i-x\|\leq (1+\epsilon)r$ and $\forall x\in X\cap B_r(c_i)$, where $B_r(c_i)=\{x: x\in\mathbb{R}^d, \|x-c_i\|\leq r\}$ we have that $x\in P_i$ in time $\tilde{O}(dn+n^{2-\alpha})$, where $\alpha = \Omega(\frac{\epsilon^{\frac{1}{3}}}{\log(1/\epsilon)})$.
\end{lemma}
\begin{proof}
We reduce the problem to Hamming space with error $\epsilon$ and radius $r$, applying the same techniques as in Theorem \ref{r-net in l1} for $\ell_1$ points or Theorem \ref{l2 to l1} for points in $\ell_2$. After this reduction, we get a set of new points $X'$ and a new radius $r'$. We apply brute force on $X'$ to get part of the solution. We randomly choose a point of $c_i\in C$ and then iterate over every point in $x\in X'$. We check if $\|x-c_i\|\leq r'+\epsilon k$, for $k=(\epsilon^{-2}\log{n})$ the dimension in Hamming space. For every point where this holds, we the original point into the set $P_i$ and then delete $x$ from $X'$. We do this $\sqrt[]{n}$ times which takes $\tilde{O}(n^{1.5})$ time in total. Without loss of generality, we now assume that $|C|>\sqrt[]{n}$, as otherwise we would be done at this point. With a similar argument as in Theorem \ref{brute force}, we argue that, after using brute force, with probability at least $1-n^{-0.2}$, $|\{(c,x)|c\in C, x\in X', \|x-c\|\leq r'+\epsilon k\}|\leq n^{1.7}$. As in Theorem \ref{hamming r-net}, we then build the distance matrix $W$ of the remaining points in $X'$. We iterate over every column corresponding to remaining center points $c_j$. For every entry $W_{i,j} > 2|S_i|$, we add original version of every point $x\in S_i$ such that $\|x-c_j\|\leq r' +k\epsilon$ to $P_j$. This takes time $\tilde{O}(n^{2-\alpha}+n^{1.7})$. It then holds that $\forall c_i\in C, \forall x \in P_i, \|c_i-x\|\leq (1+\epsilon)r$, as we only added points to $P_i$'s, where this property holds. It also holds that $\forall x\in X\cap B_r(c_i)$, as every point is only within the ball of one single center, because of the constraint on $C$. As we only do brute force and then build a distance matrix which we iterate through in a similar fashion as in Theorem \ref{hamming r-net}, the runtime is as stated.
\end{proof}

The proof of the following Theorem describes how to utilize the above Lemma to build a decider. The framework then allows us to solve a concrete instance of Min-Max Clustering. A similar decider was built by \cite{PR14}, to solve the same problem in low dimensional space.
\begin{theorem}
Let $P\subset\mathbb{R}^d$, let $(P,\mathcal{F})$ be a sketchable family and let $\|\cdot\|$ be the norm, that denotes the $\ell_1$ or $\ell_2$ norm. For a set $W\in\mathcal{F}$, let $r_{min}(W)$ be the smallest radius, such that a ball centered at a point of $W$ encloses the whole set. We can then, for $\epsilon\gg \frac{4\log^6(\log{n})}{\log^3{n}}$,  $(4+\epsilon)$-approximate the min-max clustering of $P$ with $r_{min}(W)$ as the cost function, with high probability, in time $\tilde{O}(dn+n^{2-\alpha})$, where $\alpha = \Omega(\frac{\epsilon^{\frac{1}{3}}}{\log(1/\epsilon)})$.

Specifically, one can cover $P$ by a set of balls and assign each point of $P$ to a ball containing that point, such that the set of assigned point of each ball is in $\mathcal{F}$ and the maximum radius of these balls is a $(4+\epsilon)$-approximate of the minimum of the maximal radius used by any such cover.
\end{theorem}

\begin{proof}
First notice that, when building a $(1+\epsilon)$-approximate $(4r_{opt}(1+\epsilon))$-net, where $r_{opt}$ is the radius of the optimal clustering $\mathcal{P}_{opt}$ of $P$, the following properties hold. Let $W_i\in\mathcal{P}_{opt}$ be the cluster that contains center $c_i$ of the $r$ net. It then holds that $diam(W_i)\leq 2r_{opt}$. Also any two center points of the net have distance at least $(4r_{opt}(1+\epsilon))$ from each other, thus there are no $i\neq j$ such that $W_i = W_j$. Now define $C_i$ as the set of points that are contained in a ball of radius $2r_{opt}$ around center $c_i$, hence $C_i = P\cap B_{2r_{opt}}(c_i)$. It then holds that $W_i\subset C_i$ and since $W_i\in \mathcal{P}_{opt}$, we know that $W_i\in\mathcal{F}$. Thus $C_i\in\mathcal{F}$ by definition of upward closed set systems.

This observation allows us to build a decider for $f(P,\mathcal{F})$, which is the function returning the optimal solution to the objective of the clustering. First, we build a $(1+\frac{\epsilon}{4})$-approximate $(4r(1+\frac{\epsilon}{4}))$-net. We then apply Lemma \ref{2r balls} on inputs $X=P$,$r=2r$, $\epsilon=\frac{\epsilon}{4}$ and $C=C$ , where $C$ is the set of center points of the net. It is easy to see, that the needed property on $C$ is met,as it contains the center points of a $(4r(1+\frac{\epsilon}{4}))$-net. The sets $P_i$, that get returned by Lemma \ref{2r balls}, are then supersets of $C_i$, if $r \geq r_{opt}$. From the definition of sketchable families we know that, for every $P_i$, we are able to decide if $P_i\in\mathcal{F}$ in $O(n)$. Assume now that there exists a $P_i$ which is not in $\mathcal{F}$. $P_i$ is thus not a superset of $C_i$, and we return $r < r_{opt}$. Otherwise, we know that $\forall i,C_i \subset P_i$ and thus $r \geq r_{opt}$. Now its left to decide if $r_{opt}$ is within some constant spread interval. To that end, we repeat the process above, but for a slightly smaller net, say a $(1+\frac{\epsilon}{4})$-approximate $4r$-net. If all of the $P_i$'s for this new net are in $\mathcal{F}$, we know that the original $r$ was to big and we return $f(P,\mathcal{F})<r$. Otherwise, because we applied Lemma \ref{2r balls} to radius $\frac{2r}{(1+\frac{\epsilon}{4})}$ and found that balls of that radius centered at $c_i$ are not in $\mathcal{F}$, we know that the optimal value is at least $\frac{r}{(1+\frac{\epsilon}{4})}$, dew to our observation about the diameter of the clusters in the beginning. We also know that $r_{opt}$ can be as big as $4(1+\frac{\epsilon}{4})^2r$ as we are able to cover the whole space with balls of such radius and subsets of these are in $\mathcal{F}$. Therefor, we return the interval $[\frac{r}{1+\frac{\epsilon}{4}}, 4(1+\frac{\epsilon}{4})^3\frac{r}{1+\frac{\epsilon}{4}}]$. Plugging this into the framework thus provides us with a constant spread interval, which contains the solution. By searching the interval the same way as in Theorem \ref{kth nn dist}, we end up with an interval $[\frac{r}{1+\frac{\epsilon}{4}}, 4(1+\frac{\epsilon}{4})^3\frac{r}{1+\frac{\epsilon}{4}}]$. We return $\frac{r}{1+\frac{\epsilon}{4}}$, which is a $(4+\epsilon)$-approximate solution since the real solution may be up to a $4(1+\frac{\epsilon}{4})^3$-factor off and $4(1+\frac{\epsilon}{4})^3=4((\frac{\epsilon}{4})^3+3 (\frac{\epsilon}{4})^2+3\frac{\epsilon}{4}+1) \leq (4+\epsilon)$.
In the worst case, the decider builds an approximate $r$-net twice and also calls Lemma \ref{2r balls} twice. Applying the framework with that decider and searching the returned interval thus results in $\tilde{O}(dn+n^{2-\alpha})$.
\end{proof}

\subsection{$k$-center}
The $k$-center clustering is tightly coupled to the problem of building $r$-nets. For a set of high dimensional points, we want to find $k$ clusters, that minimize the maximum diameter of any of these. For any $\epsilon>0$, computing a $(2-\epsilon)$ approximate $k$-center clustering in polynomial time has been shown to be impossible except $P=NP$ \cite{HSU}. We thus focus on computing $(2+\epsilon)$-approximates of the optimal solution. In the following we present two approaches to this. First, we build a decider, such that we are able to employ the framework, which provides us with a $(4+\epsilon)$-approximate $k$-center clustering. We then exhibit a different approach to the problem. Instead of relying on the framework, we derive an algorithm that computes approximate greedy permutations. We then present a way of reducing the computation of a $(2+\epsilon)$-approximate $k$-center clustering to building approximate greedy permutations. The drawback of this approach is, that the runtime has a logarithmic dependency on the spread of the data.

\subsubsection{$(4+\epsilon)$ approximate $k$-center}
As in previous subsections, we design a decider which then gets called by the framework. The construction is similar to \cite{PR14}, where they construct a decider to $(4+\epsilon)$-approximate $k$-center clustering in lower dimensions.

We first proof the following Lemma, which is going to be useful later.

\begin{restatable}{lemma}{kcenterlemma}
\label{k-center/r-net lemma}
There are the following relations between a set $C$, which contains the net points of a $(1+\epsilon)$-approximate $r$-net on a set of points $X$, and the function $f(X,k)$, which returns the optimal clustering radius for the $k$-center problem on the set $X$.
\begin{enumerate}
\item If $|C| \leq k$ then $f(X) < (1+2\epsilon)r$
\item If $|C| > k$ then $r \leq 2f(X)$
\end{enumerate}
\end{restatable}

\begin{theorem}
For a set of $n$ points $X\in\mathbb{R}^d$, some integer $k$, $n \geq k > 0$, some error parameter $\epsilon\gg \frac{32\log^6(\log{n})}{\log^3{n}}$ and the norm $\|\cdot\|$, that denotes the $\ell_1$ or $\ell_2$ norm, with high probability, we return a $(4+\epsilon)$-approximate $k$-center clustering in time $\tilde{O}(dn+n^{2-\alpha})$, where $\alpha = \Omega(\frac{\epsilon^{\frac{1}{3}}}{\log(1/\epsilon)})$.
\end{theorem}

\begin{proof}
In the following we present the decider that we plug into the framework. First we create a $(1+\frac{\epsilon}{32})$-approximate $\frac{r}{1+\frac{\epsilon}{16}}$-net and check if we get $k$ or less center points. If we do, then due to Lemma \ref{k-center/r-net lemma} we can safely return $f(X,k)<r$. If we do not, we create a $(1+\frac{\epsilon}{32})$-approximate $(2(1+\frac{\epsilon}{16})r)$-net and check if we have at most $k$ centers in this net. In that case, due to Lemma \ref{k-center/r-net lemma} we know that $\frac{r}{2(1+\frac{\epsilon}{16})}\leq f(X,k) < 2(1+\frac{\epsilon}{16})^2r$ and we return this interval. Otherwise, we have more than $k$-centers. Thus we know from Lemma \ref{k-center/r-net lemma} that $r(1+\frac{\epsilon}{16})\leq f(X,k)$ and  we return $f(X,k) > r$. We therefor satisfy the properties of a $(1+\epsilon)$-decider and apply it to the framework to compute a constant spread interval containing the exact solution. As in the previous subsections, we slice the interval and do binary search using the decider. If we find an interval $\frac{r}{2(1+\frac{\epsilon}{16})}\leq f(X,k) < 2(1+\frac{\epsilon}{16})^2r$, we return $2(1+\frac{\epsilon}{16})^2r$ which is a $(4+\epsilon)$ approximation, as it might miss the real solution up to a factor of $4(1+\frac{\epsilon}{16})^3=4+\frac{3\epsilon}{4}+\frac{3\epsilon^2}{64}+\frac{\epsilon^3}{1024}\leq 4+\epsilon$. It is easy to see that the decider runs in time $\tilde{O}(dn+n^{2-\alpha})$ and as in the previous subsection, the search takes $O(1/\epsilon^2)$ iterations. Therefor, the runtime is as stated.
\end{proof}

\subsubsection{$(2+\epsilon)$ approximate $k$-center with dependency on the spread}
Another way to approach the approximate $k$-center clustering, is given by \cite{Gon85}. There they construct a $2$-approximate $k$-center clustering by using greedy permutations. A greedy permutation of a point set is an ordering, such that the $i$-th point is the furthest from all previous points in the permutation. In \cite{EHS15} they describe a way of constructing an approximate greedy permutation by building approximate $r$-nets. In the following, we present a way to improve this construction by applying the approximate $r$-net algorithm from the previous section. We then present how to exploit approximate greedy permutations to create $(2+\epsilon)$-approximate $k$-center clusterings. Unfortunately, building the greedy permutation has a logarithmic runtime dependency on the spread, which is the ratio of the biggest to the smallest distance within the point set. Therefore, the algorithm is only useful for data, where the spread is in $poly(n)$.

\paragraph{Approximate greedy permutation}
A greedy permutation is an ordered set $\Pi$ of the input points, such that the point $\pi_i$ is the furthest point in $V$ from the set $\{\pi_j\}_{j=1}^{i-1}$. The following is a formal definition of approximate greedy permutations, as described in \cite{EHS15}.

\begin{definition}
A Permutation $\Pi$ is a \textit{$(1+\epsilon)$-greedy} permutation on $n$ points on metric space $(V,d)$, if there exists a sequence $r_1\geq r_2 \geq ... \geq r_n$ s.t.
\begin{enumerate}
\item The maximum distance of a point in $V$ from $\{\pi_j\}_{j=1}^{i}$ is in the range $[r_i, (1+\epsilon)r_i]$
\item The distance between any two points $u,v\in \{\pi_j\}_{j=1}^{i}$ is at least $r_i$
\end{enumerate}
\end{definition}

We now proof the following Lemma, which helps us build the approximate greedy permutation later.

\begin{lemma}\label{fixed center rnet}
For a set of $n$ points $X\subset\mathbb{R}^d$, the norm $\|\cdot\|$, that denotes the $\ell_1$ or $\ell_2$ norm, a set of points $C\subset X$, such that $\forall x,y\in C, \|x-y\|\geq r$ , some error $\epsilon \gg \frac{\log^6(\log{n})}{\log^3{n}}$ and a radius $r\in\mathbb{R}$, with high probability, we can compute a set $F$, such that $\forall y\in F, c\in C, \|x-y\|\geq r$. We do this in time $\tilde{O}(dn+n^{2-\alpha})$, where $\alpha = \Omega(\frac{\epsilon^{\frac{1}{3}}}{\log(1/\epsilon)})$. 
\end{lemma}

\begin{proof}
We proceed similar as when building the approximate $r$-net. We first reduce the problem to Hamming space with additive error $\epsilon$ as in Theorem \ref{r-net in l1}  for points in $\ell_1$ or Theorem \ref{l2 to l1} for $\ell_2$ points. We then arrive at the mapped point set $X'$ and radius $r'$. Next, we apply a slightly modified version of Theorem \ref{brute force}. Instead of randomly choosing a point from $X'$, we randomly choose a point $c\in C$. We then iterate over every point in $x\in X'$ and check if $\|x-c\|\leq r'+\epsilon k$ for $k=(\epsilon^{-2}\log{n})$, the dimension in Hamming space. For every point $x$ where this holds, we delete $x$ from $X'$ as well as from the original set $X$. We do this $\sqrt[]{n}$ times, which takes $\tilde{O}(n^{1.5})$ time in total. We now assume, without loss of generality, that $|C|>\sqrt[]{n}$, as otherwise we would be done at this point. By applying a similar argument as in Theorem \ref{brute force}, it holds that with probability at least $1-n^{-0.2}$, $|\{(c,x)|c\in C, x\in X, \|x-c\|\leq r'+\epsilon k\}|\leq n^{1.7}$. As in Theorem \ref{hamming r-net}, we now build the distance matrix $W$ of the remaining points in $X'$. We then iterate over every column corresponding to the remaining center points $c_j$ and, for every entry $W_{i,j} > 2|S_i|$, delete the original version of every point $x\in S_i$ such that $\|x-c\|\leq r' +k\epsilon$ from $X$. This takes time $\tilde{O}(n^{2-\alpha}+n^{1.7})$. The point set $X$ then, by construction, only contains points which are further then $r$ from any point in $C$.
\end{proof}

The algorithm for building the greedy permutation is very similar to the algorithm presented in \cite{EHS15}. We build sequences of approximate $r$-nets, starting with a radius that is an approximation of the maximum distance within the point set. Then, we  consecutively build approximate $r$-nets for smaller and smaller $r$, while keeping centers of previously computed approximate $r$-nets. Putting the center points into a list in the order they get computed results in an approximate greedy permutation. For a full proof outline refer to the supplementary material.

\begin{restatable}[]{theorem}{greedypermut}\label{greedy permut}
For point set $X\subset\mathbb{R}^d$ with $|X| = n$, $\epsilon \gg \frac{4\log^6(\log{n})}{\log^3{n}}$ and the norm $\|\cdot\|$, that denotes the $\ell_1$ or $\ell_2$ norm, with high probability, we can compute a $(1+\epsilon)$-approximate greedy computation in time $\tilde{O}((dn + n^{2-\alpha})\log{\Phi})$, where $\alpha = \Omega(\frac{\epsilon^{\frac{1}{3}}}{\log(1/\epsilon)})$ and $\Phi = \frac{\max_{x,y\in X, x\neq y}\|x-y\|}{\min_{x,y\in X, x\neq y}\|x-y\|}$ is the spread of data $X$.
\end{restatable}

\paragraph{$k$-center with approximate greedy permutation}
In \cite{Gon85}, Gonzales proved that an exact greedy permutation leads to a 2 approximation of the solution for the $k$-center objective, if we take the first $k$ elements out of the permutation and declare them as cluster centers. The maximum radius of a cluster, is then the minimum distance of the $k+1$-th point in the permutation to the one of the first $k$ points. With a $(1+\epsilon)$-approximate greedy permutation we can then derive a $(2+\epsilon)$-approximate solution for the $k$-center problem, since for every element $\pi_i$ and every $r_i$ as in the definition of the approximate greedy permutation, we know that, in metric space $(V,\|\cdot\|)$, it holds that $r_i\leq\max_{u\in V}\min_{j\in\{1,...,i\}}\|\pi_j-u\|\leq(1+\epsilon)r_i$. Thus the radius used, if we take the first $k$ elements of the approximate greedy permutation as cluster centers, is at most a $1+\epsilon$ factor larger than the radius we would use by taking the first $k$ elements of the exact greedy permutation, which in turn is at most a $2$ factor larger than the exact $k$-center clustering radius.
\newline

\section{Conclusion \& Future Work}
Our work has lead to interesting improvement on the construction time of approximate $r$-nets and applications thereof. We wish to highlight the following open problems.
First, can we find a lower bound to the construction time of $r$-nets? This would also tell us more about the limits of the Net \& Prune framework.
Second, can we get rid of the spread dependency on the approximate greedy permutation algorithm, as this would make the algorithm suitable for much more general data sets? Our work seems to suggest that this is tightly coupled to finding all nearest neighbors. 
\section{Acknowledgments}
We thank Ioannis Psarros for the helpful discussions. Yuyi Wang is partially supported by X-Order Lab. 

\bibliography{references}
\bibliographystyle{aaai}
\clearpage

\appendix
\section{Proof of Theorem \ref{brute force}}
\bruteforce*
\begin{proof}
We create a copy of $X$ and call it $X'$. After that, we repeat the following $\sqrt[]{n}$ times: Choose a point $x_i \in X'$ uniformly at random, delete it from $X'$ and add it to $C$. Then check for each $x\in X'$, if $\|x-x_i\|_1 \leq r+\epsilon d$. If so, delete $x$ from $X'$ as well. We do this in $O(dn^{1.5})$ time.\\
Let $Y:=|\{\{i,j\}|x_i,x_j\in X',\|x_i-x_j\|_1\leq r+\epsilon d\}|$. We now prove $Pr[Y\leq n^{1.7}]\geq 1-n^{-0.2}$ by doing a case distinction. Let $A_i$ be the number of points with small distance to a randomly chosen point $p$ in the $i$-th iteration. 
Now first assume that $\mathbb{E}[A_i]>2n^{0.5}; \forall i\in\{1,...,\sqrt[]{n}\}$. Thus the number of points to be deleted in each iteration is at least $2n^{0.5}+1$ in expectation which results in more then $n$ deleted points after $\sqrt[]{n}$ iterations. Therefor, if our assumption holds, we get $|X'|=0$ after at most $\sqrt[]{n}$ iterations. $Pr[Y\leq n^{1.7}] \geq 1-n^{-0.2}$ then holds as $Y=0$.\\
Next assume that $\exists i\in\{1,...,\sqrt[]{n}\}, \mathbb{E}[X_i]\leq 2n^{0.5}$. If we reach such an $i$, we have at most $2n^{0.5}$ small distances left between a random point of $X'$ and all the points within $X'$ in expectation. After all iterations, the number of "small" distances is therefor no more then $2n^{1.5}$ in expectation. Thus, by Markov's inequality:
\begin{equation}
\begin{split}
Pr[Y\leq n^{1.7}] &= 1- Pr[Y\geq n^{1.7}]\geq 1-O(n^{-0.2}).
\end{split}
\end{equation}
We now prove correctness of the partial approximate r-net $C$. For every point $p\in X\setminus X'$ there will be a net point in $C$ that has distance at most $r+\epsilon d$ from $p$, as we only removed points from X' which satisfy this property. For every two points $p,q\in C$ we have $\|p-q\|_1>r$, because if the distance was less or equal to $r$, either $p$ would have been deleted in the iteration of $q$ or vice versa (whatever point came first). This concludes the proof.
\end{proof}

\section{Proof of Theorem \ref{distance matrix}}

\distancematrix*

\begin{proof}
We construct an algorithm that is similar to the one used for nearest/furthest neighbor search in Hamming space, presented in \cite{alman2016polynomial}. We first create a random partition of  $X$ into disjoint sets $S_1,...,S_{n^{1-\alpha}}$, each of size $s := n^\alpha$. For every such $S_i$ and every point $q \in X$ we then want to test, if at least one point is within $(r+\epsilon d)$ distance of $q$ or not. This can be expressed as a Boolean formula in the following way:\\
\begin{align*}
&F(S_i, q) := [\min\limits_{p\in S_i}\|p-q\|_1 \leq r + \epsilon d] \\ &= \bigvee \limits_{p\in S_i}[\sum_{j=1}^d(p_jq_j + (1-p_j)(1-q_j))\geq d - (r + \epsilon d)]\\
&=\bigvee \limits_{p\in S_i}[\sum_{j=1}^d (p_j-0.5)(2q_j-1) \geq d-(r + \epsilon d)+0.5]\\
\end{align*}
Applying Theorem \ref{prob. PTF for OR}, we construct a probabilistic PTF to express $F(S_i,q)$. We then give a bound on the maximum number of monomials, according to Theorem \ref{prob. PTF for OR}:
\begin{align*}
&s\cdot\binom{O(d)}{O((\frac{1}{\epsilon})^{\frac{1}{3}}\log(s))} \leq n^\alpha \cdot O(\frac{d}{(\frac{\alpha}{\epsilon^{\frac{1}{3}}})\log{n}})^{O(\frac{\alpha}{\epsilon^{1/3}}\log{n})}\\
& \leq n^\alpha \cdot n^{O((\frac{\alpha}{\epsilon^{1/3}})\log(\frac{d}{\alpha\log{n}}))} \ll (n^{1-\alpha})^{0.1}
\end{align*}
As stated in \cite{alman2016polynomial}, this bound also holds for the construction time of the polynomial.\\
Next we sample a polynomial $f$ from the probabilistic PTF for $F(S_i, q)$. As presented by \cite{alman2016polynomial}, we are able to do this in $O(n\log(d)\log(nd))$ time. We then split $f$ into two vectors $\phi(S_i)$ and $\psi(q)$ of $(n^{1-\alpha})^{0.1}$ dimensions over $\mathbb{R}$ s.t. their dot product results in the evaluation of the corresponding polynomial. We are able to do this as the polynomial $P(x_{11},...x_{|S_i|d})$ has parameters of the form $x_{ij}=(p_j-0.5)(2q_j-1)$.
This reduces the problem of evaluating $n^{2-\alpha}$ many polynomials to multiplying a matrix $A := \frac{n}{s}\times(\frac{n}{s})^{0.1}$, where the $i$-th row of $A$ consists of $\phi(S_i)^T$, with a matrix $B := (\frac{n}{s})^{0.1}\times n$, where the $i$-th column of $B$ consists of $\psi(x_i)$. We further reduce the multiplication, by splitting $B$ into $s$ matrices of size $(\frac{n}{s})^{0.1}\times\frac{n}{s}$. By Lemma \ref{matmult copper} we are able to do each of these multiplications in $\tilde{O}((\frac{n}{s})^2)$ arithmetic operations over an appropriate field. The total time of the multiplications is then $\tilde{O}(\frac{n^2}{s}) = \tilde{O}(n^{2-\alpha})$ as we do $s$ matrix multiplications.

We then reassemble each of the $s$ matrices by placing them next to  each other, such that the $j$-th column corresponds to the point $q_j\in X$. This leads to the matrix $W$ where

\begin{enumerate}
\item $W_{ij} > 2|S_i|$ if
\begin{align*}
&\bigvee \limits_{p\in S_i}[\sum_{j=1}^d (p_j-0.5)(2q_j-1) \geq d-(r + \epsilon d) + 0.5 + \epsilon d] \\
&= [\min\limits_{p\in S_i}\|p-q\|_1 \leq r]
\end{align*}
\item $|W_{ij}| \leq |S_i|$ if 
\begin{align*}
&\bigwedge \limits_{p\in S_i}[\sum_{j=1}^d (p_j-0.5)(2q_j-1) < d-(r + \epsilon d)+0.5] \\
&=[\min\limits_{p\in S_i}\|p-q\|_1 > r+\epsilon d]
\end{align*}
\end{enumerate}
By Theorem \ref{prob. PTF for OR}, the error probability of each entry is $\frac{1}{3}$ which can be lowered to $\frac{1}{n^3}$ by repeating $O(\log{n})$ times and taking majorities. The overall runtime is then $\tilde{O}(n^{2-\alpha})$.
\end{proof}

\section{Proof of Theorem \ref{l2 to l1}}
\elltwo*
\begin{proof}
We define a mapping from $\ell_2$ to $\ell_1$. Every $x\in X$ gets mapped to the vector $f(x) = (f_1(x),...,f_k(x))$ where $k =(\epsilon^{-2} \log{n})$ and $f_i(x) = \sum_{j = 1}^{d}\sigma_{ij}x_j$. The coefficients $\sigma_{ij}$'s are independent normally distributed random variables with mean 0 and variance 1. As presented in \cite{MAT}, it holds that for any two points $x,y\in X$, \\
$(1-\epsilon)\|x-y\|_2\leq C\|f(x)-f(y)\|_1\leq(1+\epsilon)\|x-y\|_2$ with probability $1-O(\frac{1}{n})$ for some constant $C$. The cost of applying the mapping is $O(dkn)$. We then employ Theorem \ref{r-net in l1} on the new set of points, to get a $(1+\epsilon)$-approximate $r$-net in the time stated.
\end{proof}

\section{Proof of Theorem \ref{kth nn dist}}
\kthdist*
\begin{proof}
First we describe the decider, which is basically the same as in \cite{AEKP}, except that we plug in the new algorithm for {\fontfamily{cmtt}\selectfont DelFar}. The decider then works as follows: We first call {\fontfamily{cmtt}\selectfont DelFar} on the set $X$ with radius $r/(1+\frac{\epsilon}{4})$ and error $\epsilon/4$ to get a set $W_1$. Then we call {\fontfamily{cmtt}\selectfont DelFar} on $X$ again but this time with radius $r$ and error $\epsilon /4$ to get another set $W_2$. If it then holds that $|W_1| \geq k$, we know that when drawing balls of at most radius $r$ around each point, at least $k$ of the points have their nearest neighbor within their ball. This means, that $r$ is to big and we output $f(X, k) < r$. Similar, if $|W_2| < k$, we know that even if we draw balls around all the points with at least radius $r$, not even $k$ points have their nearest neighbor inside their ball which implies that $r$ is to small and we output $f(X,k)>r$. Finally, if we have that $|W_1| < k$ and $|W_2|\geq k$, we know that the exact $k$-nearest neighbor has to be in the range $[r/(1+\frac{\epsilon}{4}),(1+\frac{\epsilon}{4})r]$ and we output that interval.\\
As this satisfies the definition of a $(1+\epsilon)$-decider, we plug it into the framework and get a constant spread interval $[x,y]$ which contains the exact solution. We then use the decider again to $(1+\epsilon)$-approximate the exact solution. We slice the interval into pieces $x,(1+\epsilon)x,(1+\epsilon)^2x,...,y$ and do binary search on those slices, by applying the decider. If we hit an $r$ where the decider gives us an interval $[r/(1+\frac{\epsilon}{4}),(1+\frac{\epsilon}{4})r]$, we return $(1+\frac{\epsilon}{4})r$ and are done. The optimal solution might then be $r/(1+\frac{\epsilon}{4})$, which is a $(1+\frac{\epsilon}{4})^2$ factor smaller then what we return. This is fine as $(1+\frac{\epsilon}{4})^2 = 1+\epsilon/2 + \epsilon^2/16\leq 1 + \epsilon$ and what we return is thus a $(1+\epsilon)$-approximate as desired. While searching we make $O(1/\log(1+\epsilon)) = O(1/\epsilon^2)$ calls to the decider. The search thus ends up having a runtime of $\tilde{O}(dn + n^{2-\alpha})$.
\end{proof}

\section{Proof of Lemma \ref{k-center/r-net lemma}}
\kcenterlemma*
\begin{proof}
For the first property we create a $(1+\epsilon)$-approximate $r$-net of $X$. Due to the covering property of approximate $r$-nets, every point in $X$ is within $(1+\epsilon)r$ of a center point and thus $(1+2\epsilon)r$ is not an optimal radius for the $k$-center clustering.

For the second property note that an approximate $r$-net with more then $k$ centers contains at least $k+1$ centers. These are at least $r$ from each other due to the packing property. Thus $k$ centers with a radius of $<r/2$ would not be able to cover all of the $k+1$ centers from the approximate $r$-net and hence $r \leq 2f(X)$.
\end{proof}

\section{Proof of Theorem \ref{greedy permut}}
\greedypermut*
\begin{proof}
In \cite{EHS15} they presented a way to compute a $(1+\epsilon)$-approximate greedy permutation. In the following, when talking about building an approximate $r$-net we refer to Theorem \ref{r-net in l1} for Euclidean points with $\ell_1$ metric or Theorem \ref{l2 to l1} for points in $\ell_2$ metric space. 

We choose a point $p$ at random and search for the furthest neighbor of that point. $\Delta$, which is the distance to the furthest neighbor, is then a $2$-approximation to the maximum distance within $X$ by the triangle inequality. We thus know that $\max_{x,y\in X, x\neq y}\|x-y\| \geq \Delta\geq\frac{1}{2} \max_{x,y\in X, x\neq y}\|x-y\|$. Next we define a sequence of radiuses $r_i = \frac{\Delta}{(1+\frac{\epsilon}{4})^{i-1}}$ for $i \in \{1,..,M:=\left \lceil\log_{1+\frac{\epsilon}{4}}{\Phi}\right\rceil+2\}$ where $\Phi := \frac{\max_{x,y\in X, x\neq y}\|x-y\|}{\min_{x,y\in X, x\neq y}\|x-y\|}$ is the spread of the data. Note that $r_M \leq \frac{\min_{x,y\in X, x\neq y}\|x-y\|}{(1+\epsilon)}$. We then iterate over this sequence, where in the first iteration, we compute an approximate $r_1$ net $C_1$ of $X$. We do not need to run the algorithm as we know that $\{p\}=C_1$ as $r_1 =\Delta$ and a ball of radius $\Delta$ around $p$ covers $X$ by construction. In every iteration $i>1$, we then first define the set $S_i = \bigcup_{j=1}^{i-1}C_j$. Next, we alter the approximate $r$-net algorithm in such a way, that the points in $S_i$ are already centers of the net and thus points within distance $r_i$ of $S_i$ are not added as net points. By Lemma \ref{fixed center rnet} we are able to do this, without an increase of the runtime. We then apply Lemma \ref{fixed center rnet} to $X$ with error $\epsilon/4$, the set $S_i$ and radius $r_i$. After that we compute a $(1+\frac{\epsilon}{4})$-approximate $r$-net of the set $F$ that gets returned. The net points are then stored in the set $C_i$. The sequence $\langle C_1,...,C_M\rangle$ then forms a $(1+\epsilon)$-greedy permutation, as shown in \cite{EHS15}. It is sufficient to do $M$ iterations as we know that a $(1+\frac{\epsilon}{4})$-approximate $r_M$-net adds all the remaining points to $C_M$ as $r_M$ is less then the minimum distance of the set.
The number of iterations is $M$ and $O(\log_{1+\frac{\epsilon}{4}}{\Phi}) = O(\frac{4}{\epsilon}\log{\Phi}) = \tilde{O}(\log{\Phi})$. In each iteration we apply Lemma \ref{fixed center rnet} and compute an approximate $r$-net, the total runtime is thus as stated.
\end{proof}


\end{document}